\newif\ifconference
\newif\ifanonymous
\conferencefalse
\anonymousfalse

\documentclass[11pt,a4paper,final]{article}

\usepackage[utf8]{inputenc}
\usepackage{lmodern}
\usepackage[T1]{fontenc}
\usepackage[english]{babel}
\usepackage{amsmath,amsfonts}
\usepackage{amssymb}
\usepackage[showonlyrefs]{mathtools}
\usepackage{amsthm}
\usepackage{moresize}
\usepackage[dvipsnames,svgnames]{xcolor}
\ifconference
\else
\usepackage[a4paper, portrait, left=1in, right=1in, top=1in, bottom=1in]{geometry}
\fi
\usepackage[backend=biber,style=numeric,sorting=nty,urldate=iso,maxnames=6]{biblatex}         % citation style AUTHOR (YEAR), or AUTHOR [NUMBER]
\usepackage[hyperindex,breaklinks,bookmarks]{hyperref}
\hypersetup{
    unicode=true,          % non-Latin characters in Acrobat’s bookmarks
    colorlinks=true,        % false: boxed links; true: colored links
    linkcolor=red,          % color of internal links (change box color with linkbordercolor)
    citecolor=DarkGreen,        % color of links to bibliography
    filecolor=magenta,      % color of file links
    urlcolor=cyan,           % color of external links
	bookmarksdepth=2
}
\usepackage[shortlabels]{enumitem}
\usepackage[toc]{appendix}
\usepackage{microtype}
\usepackage[ruled,linesnumbered]{algorithm2e}
\usepackage{aliascnt}
\usepackage[capitalize, nameinlink,noabbrev]{cleveref}
\usepackage{mdframed}
\usepackage[left,mathlines]{lineno}

\usepackage{setspace}

\ifanonymous
\linenumbers
\fi
\newmdenv[
  topline=false,
  bottomline=false,
  rightline=false,
  skipabove=\topsep,
  skipbelow=\topsep
]{lined}

\renewcommand{\O}{O}

\usepackage[textwidth=21mm,obeyFinal]{todonotes}
\let\oldtodo\todo
\renewcommand{\todo}{\textcolor{red}{Please use ``backslash your name'' instead of todo}}
\newcommand{\mytodo}[2]{\oldtodo[size=\tiny, color=#1!50!white]{#2}\xspace}

\newcommand{\jakub}[1]{\mytodo{red}{Jakub: #1}}
\newcommand{\vasek}[1]{\mytodo{blue}{Vasek: #1}}
\newcommand{\rh}[1]{\mytodo{brown}{RH: #1}}

\let\citet\textcite
\newtheorem{theorem}{Theorem}

\newtheorem{definition}{Definition}
\newaliascnt{lemma}{theorem}
\newtheorem{lemma}[lemma]{Lemma}
\aliascntresetthe{lemma}

\crefname{lemma}{Lemma}{Lemmas}

\renewcommand{\d}{\hat d}
\newcommand{\Astar}{A\kern -.1em*}

\makeatletter
\AtBeginDocument{
  \let\oldlog\log
  \let\oldlg\lg
  \def\log{\texorpdfstring{\log}{log }}
  \def\lg{\texorpdfstring{\lg}{log }}
  \hypersetup{
    pdftitle = {\@title},
    pdfauthor = {\ourauthors},
	pdfkeywords = {bidirectional search, instance-optimality, Dijkstra's algorithm},
    pdfsubject = {\ourabstract}
  }
  \let\log\oldlog
  \let\lg\oldlg
}
\makeatother

\newcommand{\crefline}[1]{line~\ref{#1}}

\date{}
\title{Bidirectional Dijkstra's Algorithm is Instance-Optimal}

\ifanonymous
\def\ourauthors{Anonymous authors}
\else
\def\ourauthors{Bernhard Haeupler, Richard Hladík, Václav Rozhoň, Robert E. Tarjan, Jakub Tětek}
\fi

\author{
	\ifanonymous
	Anonymous Authors
	\else
Bernhard Haeupler\thanks{INSAIT, Sofia University ``St.~Kliment Ohridski'' \& ETH Zurich, \texttt{bernhard.haeupler@insait.ai}}
\and
Richard Hladík\thanks{ETH Zurich, \texttt{ethz@rihl.cz}}
\and
Václav Rozhoň\thanks{Charles University, \texttt{vaclavrozhon@gmail.com}}
\and
Robert E.~Tarjan\thanks{Princeton University, \texttt{ret@cs.princeton.edu}}
\and
Jakub Tětek\thanks{\texttt{j.tetek@gmail.com}}
\fi
 }

\def\ourabstract{%
Although Dijkstra's algorithm has near-optimal time complexity for the problem of finding a shortest path from a given vertex $s$ to a given vertex $t$, in practice other algorithms are often superior on huge graphs. 
A~prominent example is \emph{bidirectional search}, which concurrently executes Dijkstra's algorithm forward from $s$ and backward from $t$, and stops when these executions meet.
\texorpdfstring{\par\smallskip\par}{}%
In this paper, we give a strong theoretical justification for the use of bidirectional search to find a shortest $st$-path. We prove that for weighted multigraphs, both directed and undirected, a careful implementation of bidirectional search is instance-optimal with respect to the number of edges it examines. That is, we prove that no correct algorithm can outperform our implementation of bidirectional search on \emph{any single instance} by more than a constant factor.
\texorpdfstring{\par\smallskip\par}{}%
For unweighted graphs, we show that bidirectional breadth-first search is instance-optimal up to a factor of $O(\texorpdfstring{\Delta}{Δ})$ where $\texorpdfstring{\Delta}{Δ}$ is the maximum degree of the graph. We also show that this is best possible. 
}

\begin{document}
\maketitle

\begin{abstract}
	\ourabstract
\end{abstract}

\section{Introduction}

From a theoretical perspective, Dijkstra's algorithm, with its near-linear time complexity, is close to optimal for the shortest $st$-path problem, that of finding a shortest path from $s$ to $t$.
When the input graph is huge, however, other algorithms for the problem are often significantly more efficient in practice than Dijkstra's algorithm. One notable algorithm for this problem is \emph{bidirectional search}, proposed by \citet{dantzig1963} in 1963 and \citet{nicholson1966finding} in 1966.  It runs Dijkstra's algorithm forward from $s$ and backward from $t$, and halts when the two executions meet. In the worst case, this algorithm has to examine all the vertices and edges of the input graph, but in practice it often examines only a small part of the graph.

In this paper, we explain why: On weighted multigraphs with positive weights, a version of bidirectional search is \emph{instance-optimal} in terms of the number of edges that the algorithm accesses. 
This means that the algorithm is the most efficient one for \emph{every single instance}.
Concretely, up to a constant factor there is no correct algorithm that accesses fewer edges than bidirectional search on even a single input. This result holds for the adjacency list model of sublinear algorithms, in which we assume that we have only simple query access to the vertices and edges of the input graph and we are not given any additional information about the graph. The time complexity of bidirectional search is proportional to the number of edges accessed, up to a logarithmic factor. Thus bidirectional search is also close to being instance-optimal from the classical time-complexity perspective. We also prove a similar result for the class of unweighted graphs and bidirectional breadth-first search (BFS).

\begin{theorem}[Informal corollary of \cref{thm:main_theorem_weighted,thm:main_unweighted,thm:lower-bound-minima}]
\label{thm:main_informal}
    In the adjacency list model of sublinear algorithms, there is an instance-optimal algorithm for the shortest $st$-path problem on weighted multigraphs.
    On \emph{unweighted} multigraphs, there is an algorithm that is instance-optimal up to a factor of $\Delta(G)$\footnote{We use $\Delta(G)$ or just $\Delta$ to denote the maximum degree of $G$. }; this factor cannot be improved.  
\end{theorem}

% Curiously, we show that the same is not the case for unweighted graphs. In unweighted graphs we show that the algorithm is instance-optimal up to a factor of the maximum degree $\Delta$. We also show that this is the best possible: There is a graph class with maximum degree $\Delta$ on which no algorithm that would be instance-optimal up to a factor of $o(\Delta)$ exists.

The proof of \cref{thm:main_informal} is technically straightforward.  But we believe that the result is appealing, for a number of reasons. 
First, we stress that our proved guarantee for bidirectional search is very strong. The result makes it clear that bidirectional search is asymptotically the best possible algorithm for the shortest $st$-path problem on \emph{any} given input graph, provided that we are not given any additional information about the graph. 

Second, there are several standard variants of bidirectional search (see \cref{sec:related_work}). We know of three variants proposed in the literature \cite{dantzig1963,nicholson1966finding,pohl1969bi}, but none is instance-optimal.
Our proof of \cref{thm:main_informal} is for a different but very simple variant of the algorithm.

Third, our result expands the set of problems that are now known to have instance-optimal algorithms. Although there are several areas of computer science for which results of this type are known, it seems that in the setup of \cref{thm:main_informal} -- sublinear graph algorithms -- our result is the first of its kind. We discuss other problems amenable to this type of analysis in \cref{sec:related_work}. 

Finally, the shortest $st$-path problem is interesting from the perspective of determining what conditions are required to achieve instance optimality.
In particular, our results show that certain subtle assumptions about edge weights are crucial. Without these assumptions, instance optimality becomes impossible to attain. This highlights the delicate choices one sometimes has to make to prove instance optimality. We discuss this in \Cref{subsec:remarks}.

\paragraph{Optimality of the unidirectional Dijkstra algorithm in a restricted setting} Additionally, we prove in \Cref{sec:unidirectional_dijkstra} that on directed graphs, assuming we can access only the outgoing edges of vertices and that we cannot access the vertex degrees, the unidirectional Dijkstra algorithm is instance-optimal. The proof is very simple, and it showcases the general approach that we use in our other proofs. For the sake of simplicity of exposition, we only prove this claim for the class of deterministic algorithms.  We prove our other results for randomized  algorithms.

Although the fact that our lower bounds hold for randomized as well as deterministic algorithms may seem uninteresting at first, we think that the opposite is true. 
Although the classical shortest-path algorithms, those of Dijkstra \cite{Dij59} and Bellman-Ford \cite{bellman1958routing,schrijver2005history} are deterministic, many state-of-the algorithms for the shortest path problem in various models of computation are, in fact, randomized. 
This includes the fastest known algorithm for computing distances in undirected graphs~\cite{duan2023randomized} and 
the state-of-the-art parallel algorithms for the shortest path problem~ \cite{cao2020improvedtradoffs,cao2020paralleldirected,cao2023parallel,rozhovn2023parallel}. 
We also note that the sublinear model of computation is notorious for requiring randomization in order to get any non-trivial algorithm for most problems.

\section{Related work}
\label{sec:related_work}
This section presents an overview of related work.
First, we explain the bidirectional search meta-algorithm and its variants proposed in the literature. Next, we discuss the connections with the popular \Astar{} algorithm, which solves the shortest $st$-path problem assuming we have additional information about distances to $t$. Finally, we discuss other setups and problems that admit instance-optimality-based analysis. 

\subsection*{The bidirectional Dijkstra algorithm} This algorithm was first proposed by \citet{dantzig1963} but his description was very vague. The first to precisely specify a correct algorithm was \citet{nicholson1966finding}.
The general algorithm is non-determinstic in that it can choose at each step to do a step of the forward search or a step of the backward search.  We formulate the core structure of the algorithm as a meta-algorithm in \cref{alg:metaalgorithm}.

\begin{algorithm}
\SetAlgoLined
\KwIn{Graph $G(V,E)$, source vertex $s$, target vertex $t$}
\KwIn{Selection function, stopping condition}
Initialize forward search from $s$ and backward search from $t$\;
\While{stopping condition is not satisfied}{
    Use a selection rule to select one of the two directions\; 
    \eIf{direction is forward}{
        Do one edge relaxation of Dijkstra's algorithm from $s$\;
    }{
        Do one edge relaxation of reverse Dijkstra's algorithm from $t$\;
    }
}
Recover the shortest path from the two runs;
\caption{Bidirectional Search Meta-Algorithm}
\label{alg:metaalgorithm}
\end{algorithm}

\paragraph{Selection rule} Several rules have been proposed for selecting between the two searches in \cref{alg:metaalgorithm}.  All of them process one vertex at a time, choosing as the next vertex either the next vertex to be processed by the forward search or that to be processed by the backward search.  Processing a vertex relaxes all its incident relaxable edges in the appropriate direction (outgoing or incoming). \citet{dantzig1963} suggests alternating between processing a vertex in the forward search and processing one in the backward search. 
\citet{nicholson1966finding}, on the other hand, suggests always processing a vertex $v$ with $\min\{\d(s, v), \d(v, t)\}$ minimum, where $\d(x, y)$ is the minimum distance from $x$ to $y$ computed so far. 
Finally, \citet{pohl1969bi} suggests processing a vertex in the execution that has reached but not processed the smaller number of vertices. 
None of these methods achieves the instance-optimality guarantees of \cref{thm:main_informal}. In our instance-optimal \cref{alg:instanceoptimal} we shall use a different yet very simple equal-work rule, that of alternating between relaxing an edge in the forward and backward searches.

\paragraph{Stopping condition}
There are also differences in the stopping condition. \citet{dreyfus1969appraisal} proposes stopping when the two executions meet (for a correct definition of ``meets''). \citet{pohl1969bi} proposes a stopping condition that uses distance bounds computed by the algorithm to determine that a shortest path has been found. We use the latter idea.

There is a seemingly natural stopping condition that is not correct: if we stop when the two sets of reached vertices intersect for the first time, the vertex at which the two executions meet may not lie on the shortest path.
Indeed, this incorrect stopping rule appears in the literature, as was pointed out by \citet{pohl1969bi}. 

\paragraph{Optimality} We are not aware of works that give theoretical guarantees of bidirectional search comparable to \cref{thm:main_informal}. In \cite{pohl1969bi}, Pohl gives a heuristic argument. Namely, he models the input graph in a continuous probabilistic way and argues heuristically that under this model, his version of bidirectional search is the best possible. %However, he also points out that on some inputs, this algorithm does not have to be optimal (although he only gives one such graph as an example, which technically does not preclude the algorithm from being instance-optimal up to a constant).
In general, it is known that bidirectional search or algorithms based on it can have sublinear time complexity for certain specific classes of graphs such as hyperbolic random graphs \cite{blasius2022efficient}, power-law graphs \cite{borassi2019kadabra}, Chung-Lu random graphs  \cite{basu2024sublinearalgorithmapproximateshortest}, and expanding graphs \cite{alon2023sublinear,blasius2023deterministic}. 

\subsection*{Relation to the \Astar{} algorithm}
One might wonder how it is that bidirectional Dijkstra is instance-optimal, but in practice it is often outperformed by the famous \Astar{} algorithm. 
The reason is that \Astar{} requires additional advice (heuristic distance estimates) as a part of the input. 
More formally, \Astar{} solves a different problem than bidirectional Dijkstra; specifically, the input consists of both the graph and an estimate for every vertex of its distance to the destination. 
This shows an important nuance of \cref{thm:main_informal}: Bidirectional search is optimal only if we assume that we have no way of learning additional meaningful information about the input graph. 

The \Astar{} algorithm was first suggested by \citet{hart1968}. The algorithm was later shown to be optimal among all unidirectional-Dijkstra-like algorithms by \citet{dechter1985}. 

A bidirectional version of the \Astar{} algorithm was suggested by \citet{holte2017}. Different versions of the bidirectional \Astar{} algorithm have been studied; we refer the interested reader to \cite{holte2017,sturtevant2018,shaham2019}. \citet{eckerle2017} made significant progress in proving that bidirectional \Astar{} is optimal among bidirectional-Dijkstra-like algorithms. \citet{shaham2019} in fact used these results to give algorithms that are optimal among all bidirectional-Dijkstra-like algorithms, if one optimally sets a parameter of that algorithm. One does not know the optimal parameter value in advance however, so this result falls short of giving an algorithm optimal among bidirectional-Dijkstra-like algorithms.

%This can be seen as strenghtening the results for the standard shortest $st$-path problem from \cite{dechter1985,eckerle2017,shaham2019}. The results of \citet{dechter1985} for the \Astar{} algorithm imply that Dijkstra's algorithm is optimal in terms of the number of vertices it touches among all unidirectional-Dijkstra-like algorithms. 
%Similarly, \citet{eckerle2017,shaham2019} make significant progress in proving similar results for the bidirectional search, as we discuss in \Cref{sec:related_work}. 

Although our techniques are straightforward, and similar to those in the cited papers, our \cref{thm:main_informal} is significantly stronger than previous results, because we prove optimality among \emph{all} correct algorithms, while the previous work \cite{eckerle2017,shaham2019} restricts itself to certain classes of algorithms. Moreover, we show optimality in terms of the number of edge accesses, not vertex accesses. Since the time complexity is near-linear in the number of edge accesses of the algorithm, this means that our algorithm is also near-instance-optimal in terms of its actual time complexity, and not just in the number of reached vertices. 

\subsection*{Instance optimality for other problems}

Due to its strength, instance optimality is an extremely appealing beyond-worst-case guarantee that an algorithm can have. See the relevant chapter in \citet{roughgarden_2021} for an introduction. There are several known setups in which instance-optimality-based analyses of algorithms work. 

The original paper by \citet{fagin2001first_instance_optimal} that coined the term comes from the area of greedy algorithms for retrieving data from databases. 

Several algorithms with instance-optimality-like guarantees are known for problems related to sorting, including the problem of finding the convex hull \cite{chan2017instance_optimal_hull}, computing set intersections \cite{demaine2000adaptive}, sorting with partial information \cite{supi,van2024simpler}, and sorting vertices of an input graph by their distance from a source vertex (the main problem solved by Dijkstra's algorithm) \cite{dijkstra-universal-optimality}. 

In the area of distributed algorithms, many algorithms have been proved to have a guarantee called universal optimality, which is closely connected to instance optimality \cite{garay1998sublinear,haeupler2022hop,ghaffari2022universally}. This includes algorithms for the approximate shortest path problem \cite{haeupler2021universally,goranci2022universally,rozhon_grunau_haeupler_zuzic_li2022deterministic_sssp}. 

Another area in which instance-optimal algorithms are known is the area of sequential estimation \cite{valiant2016instance,valiant2017automatic,huang2021instance}.\vasek{add the new sampling paper} There is an instance-optimal sublinear-time algorithm for computing the distance of a point from a given curve \cite{baran2004optimal}. Yet another area with known instance-optimal algorithms is that of multi-armed bandits, with works on that topic including \cite{lai1985asymptotically,chen2016open,chen2017towards,li2022instance,kirschner2021asymptotically}.

\subsection*{Errata and Subsequent Work}

Previous versions of this paper (arXiv v1--3 and SOSA 2025 proceedings) had a mistake in \cref{alg:instanceoptimal} that led to it not being instance-optimal (albeit still correct), and a corresponding gap in the proof of \cref{thm:main_theorem_weighted}. This was pointed out independently by \citet{simplebidir-barc} and \citet{simplebidir-pozar}; we thank the authors of both papers.

Concretely, the update of $\mu$ on \crefline{line:mu_update_condition} was guarded by an additional check that the other endpoint of the relaxed edge had already been closed by the opposite execution. Because of this check, the algorithm sometimes failed to update $\mu$ (the upper bound on the $st$-distance used in the stopping condition) and kept running after it was already clear that no shorter path could be found; both papers give explicit instances on which this costs a factor of $\Theta(n)$. The fix, proposed in both \cite{simplebidir-barc} and \cite{simplebidir-pozar}, is simply to remove the additional check. \Cref{sec:weighted} now presents the corrected algorithm and a corrected proof. We also improved the readability of the proof in this version.

Both papers also study the shortest $st$-path problem on \emph{simple} graphs, which the previous versions of this paper left as an open problem. \citet{simplebidir-barc} resolve this question almost completely and show that the multigraph assumption cannot be dropped in general: while bidirectional Dijkstra remains instance-optimal on simple undirected graphs when incidence lists are in random order, it is not instance-optimal on simple undirected graphs with adversarially ordered incidence lists, nor on simple directed graphs. They further show that in almost all of these settings no algorithm can do better, so bidirectional Dijkstra remains the best possible algorithm in terms of its instance-optimality ratio even where it is not instance-optimal. \citet{simplebidir-pozar} gives complementary positive results, proving instance optimality on simple graphs for restricted classes of instances.

Finally, \citet{simplebidir-pozar} argues that the unidirectional result in \cref{sec:warmup} (\cref{thm:unidirectional_dijkstra}) is also incorrect. We believe this is a difference in models rather than an error. \Cref{thm:unidirectional_dijkstra} concerns a restricted model in which each edge $uv$ can be accessed only from its tail $u$; in particular, there is no way to enumerate the in-edges of a vertex. The counterexample of \cite{simplebidir-pozar} is given for undirected graphs, where the neighbors of $t$ are accessible from $t$, and the competing algorithm there relies on scanning the edges entering $t$ in order to lower-bound the $st$-distance. In the model of \cref{thm:unidirectional_dijkstra} this is not possible, and the theorem stands.

\section{Preliminaries}
\label{sec:preliminaries}
In this section, we review standard concepts necessary for the later analysis. 
\paragraph{Our model of sublinear algorithms}
We work within a standard query model for sublinear graph algorithms \cite[Chapter 10]{goldreich2017introduction}. The complexity measure of interest is then the number of queries performed.
We assume that the input is a weighted multigraph $G$ with $n$ vertices whose edges have real-valued weights.  We allow both self-loops and parallel edges.

We assume that the graph $G$ is stored in the incidence list format. That is, we have query access to its vertices, each of which keeps a list of its outgoing edges and a list of its incoming edges; if the graph is undirected, these lists are the same.  Each edge on such a list stores its weight and both ends of the edge.  Each edge is on two incidence lists.  
We assume that the vertices are numbered $1$ to $n$ and that vertices are identified by their number.  On an undirected graph, we are allowed to perform the following queries:

\begin{enumerate}
    \item \textbf{Degree}($i$): Given a vertex $i$ with $1 \le i \le n$, this function returns the degree of the vertex. 
    \item \textbf{Edge}($i, j$): Given a vertex $i$ with $1 \le i \le n$ and the index $j$ of its $j^{th}$ incident edge, with $1 \le j \le \deg_G(i)$\footnote{We use $\deg_G(i)$ to denote the degree of vertex $i$ in graph $G$. }, this function returns the $j$-th incident edge of $i$, along with its weight and its other end.
\end{enumerate}

On directed graphs, we can similarly query the \textbf{Indegree} and the \textbf{Outdegree} of a vertex; moreover, we have functions \textbf{Inedge} and \textbf{Outedge} that can access the incident incoming and outgoing edges of a given vertex.%\footnote{We remark that the instance optimality can also be proven in slightly stronger models allowing additional queries. One such query is \textbf{Edges}($i_1, i_2,j$) which returns the $j$-th multiedge between the two nodes $i_1, i_2$.} 

Each query has a unit cost.  We use the number of queries as our main measure of algorithmic complexity. 
This measure of \emph{query complexity} is closely related to the classical time complexity, since classical algorithms based on Dijkstra's algorithm run in time proportional to the number of edges examined, up to a logarithmic factor required for heap maintenance. 

% In the case of directed graphs, we adapt this model as follows: for each vertex, we have two separate arrays: one for its in-neighbors, and one for its out-neighbors. We may query both the indegree and outdegree of a vertex and if we have a vertex $v$ and an index $i$, we may query the $i$-th in-/out-neighbor of $v$.

We will prove instance optimality in the more general setting of randomized algorithms. For such algorithms, we count the expected number of queries. A randomized algorithm is said to be correct if it is correct on every instance with a probability at least 0.9. 

\paragraph{The shortest $st$-path problem} In the shortest $st$-path problem, we are given an input (directed or undirected) multigraph $G$. The edges of the graph have weights given by function $\ell : E(G) \rightarrow \mathbb{R}_{>0}$; this function induces a distance function $d: V(G) \times V(G) \rightarrow  \mathbb{R}_{\ge 0}$ that maps any two different nodes $u,v$ to their positive shortest-path distance $d(u,v)$ (that might be different from $d(v,u)$ on directed graphs). It is important to consider edges with strictly positive, instead of nonnegative, weights. We discuss the difference later in \cref{subsec:remarks}. 

For the shortest $st$-path problem, we are moreover given two vertices $s$ and $t$ of $G$. The task is to return an arbitrary shortest $st$-path. Note that in the case of multi-graphs, it does not suffice to just output the sequence of vertices but one must output the specific edges.

\paragraph{Dijkstra's algorithm} We now recall Dijkstra's algorithm \cite{Dij59} and define terms that we will later use in our proofs.
Throughout the execution, we store for each vertex $w$ the length $\d(s, w)$ of the shortest $sw$-path found so far. The algorithm starts with the vertex $s$ being \emph{open} and all others being \emph{unvisited}. It then repeatedly takes the vertex $u$ closest to $s$ among all open vertices and \emph{relaxes} all edges $u v$ leaving this vertex as follows. Relaxing an edge to an unvisited vertex $v$ makes $v$ open and sets $\d(s,v) \coloneq \d(s,u) + \ell(uv)$. 
If $v$ is open, we update the length of the currently shortest path to $v$ by setting $\d(s,v) \coloneq \min(\d(s,v), \d(s,u) + \ell(uv))$. If $v$ is closed, we do nothing.
Once we have relaxed all edges leaving a vertex, the vertex is \emph{closed} from that point on, and we continue with the next vertex in the list of open vertices. It should be noted that the names unlabeled, labeled, and scanned are also used for unvisited, open, and closed, respectively.

The classical Dijkstra's algorithm finishes once the target node $t$ becomes closed, or when the set of open nodes becomes empty (in which case $t$ is unreachable). However, in bidirectional search, the stopping condition is more complex as there are two runs of Dijkstra's algorithm involved. 

Once the algorithm finishes, one can recover the shortest path to any closed vertex $v$. Namely, we start with $v$ and we find a vertex $v'$ such that $\hat{d}(s,v') = \hat{d}(s,v) - \ell(v',v)$. By iteratively finding the preceding vertex like this, we may recover the whole shortest $sv$-path.

\paragraph{Instance optimality} Intuitively speaking, an algorithm is instance-optimal (up to $c$) if it is optimal (up to $c$) on every single instance among the set of correct algorithms.

We say that an \emph{algorithm $A$ is correct} if on any input, it outputs a correct output with probability $\geq 0.9$. Note that by standard probability amplification, the constant $0.9$ is arbitrary and any constant $>1/2$ would work. \vasek{clear for s-t-distance, not so clear for s-t-path?}

\begin{definition} \label{def:instance_optimality}
A correct algorithm $A$ is instance-optimal under complexity function $T$ if there exists $c = \O(1)$ such that on every input~$x$ it holds for every correct algorithm $A'$ that the expected complexity $T_A(x)$ and $T_{A'}(x)$ of respectively $A$ and $A'$ on $x$ satisfy
\[
T_{A}(x) \leq c \cdot T_{A'}(x) \,.  
\]
\end{definition}

The definition of instance optimality readily generalizes to algorithms that are instance-optimal up to a potentially nonconstant factor, such as $\Delta(G)$ that we encountered in \cref{thm:main_informal}. 

Throughout the paper, we will assume $T$ to be the query complexity.
However, note that the time complexity can be off only by a factor of at most $O(\log n)$ required to handle the heap operations necessary in Dijkstra's algorithm. 

\section{Warm-up: Unidirectional Search in Directed Graphs} \label{sec:unidirectional_dijkstra}\label{sec:warmup}
In this section, we show the simplest result of the kind that we focus on. Namely, we consider a more restricted query model, deterministic algorithms, and we focus on only computing the distance from $s$ to $t$, instead of actually finding the path. We then show that standard Dijkstra's algorithm is instance-optimal if stopped at the right time. Unlike in other sections, we also phrase the result here in a self-contained way to make it more accessible to a casual reader.

The proof is conceptually similar to other proofs in this paper, which however need to take care of several additional obstacles. We hope this proof may serve as a warm-up for the other proofs.

\vasek{is it possible this is folklore?}
\begin{theorem} \label{thm:unidirectional_dijkstra}
Let us have a directed weighted graph $G$ with positive weights and assume we are given two vertices $s,t$. Assume the only operation we can do is to take a vertex we have seen and ask for its next out-neighbor (in an adversarial ordering) and the weight of the edge to that vertex.

	Consider executing Dijkstra's algorithm from $s$ and stopping it once we close some vertex $v$ with $\hat{d}(s,v) = \hat{d}(s,t)$ (possibly $v = t$). Then this algorithm correctly computes the $st$-distance. Furthermore, no correct deterministic algorithm $A$ can perform fewer queries on $G$.
\end{theorem}
\begin{proof}
	First, we argue correctness. By the standard proof of correctness of Dijkstra's algorithm, once we close the vertex $v$, we have $\hat{d}(s,v) = d(s,v)$. At the same time, vertices are closed in order of non-decreasing distance, meaning that $d(s,t) \ge d(s, v) = \hat d(s, v) = \hat d(s, t)$. Moreover, it always holds that $\hat{d}(s,t) \geq d(s,t)$. Thus, we have $\hat{d}(s,t) = d(s,t)$, meaning that the distance is correct.

For the sake of contradiction, let us have an algorithm that performs fewer queries than Dijkstra on $G$. Therefore, there has to be an edge $uv$ for $d(s,u) < d(s,t)$ that $A$ does not query.\jakub{Why?} We define a graph $G'$ where we replace the edge $uv$ by $ut$ with weight $\delta < d(s,t) - d(s,u)$. The distance between $s$ and $t$ in $G'$ is then $d(s,u) + \delta < d(s,t)$. However, the algorithm does not query this edge. Since the rest of the graph is exactly the same, the algorithm thus returns the same answer on both $G$ and $G'$, which implies that the algorithm is not correct.
\end{proof}

\section{Instance Optimality in Weighted Graphs}
\label{sec:weighted}

Now we give an instantiation of the bidirectional search meta-algorithm that we later prove is instance-optimal. 
As the selection rule that chooses between the two executions, we use the perhaps simplest possible rule in which we alternate one edge relaxation in the forward algorithm with one edge relaxation of the backward algorithm. This way, we make sure that at any point in time, each of the two executions has the same amount of work being invested in it.

As the stopping condition, we use the classical stopping condition of \citet{pohl1969bi} in \crefline{line:break_condition}. That is, we are keeping track of the length $\mu$ of the currently shortest found path. Moreover, we are keeping track of the values $d(s, u_s), d(u_t, t)$: distances of the vertices that are currently being explored from $s$ and $t$. Once we know that $d(s,u_s) + d(u_t,t) \ge \mu$, we may terminate our search since all the paths between $s$ and $t$ we have not considered yet have to consist of two disjoint parts of lengths at least $d(s,u_s)$ and at least $d(u_t, t)$.

We present the algorithm formally in \cref{alg:instanceoptimal}.

\begin{algorithm}
\SetAlgoLined
\SetKwFunction{Fwd}{Forward\_algorithm}
\SetKwFunction{Bwd}{Backward\_algorithm}
\SetKwProg{Fn}{Function}{:}{}

\KwIn{Graph $G(V,E)$, source vertex $s$, target vertex $t$}
$\mu \gets +\infty$ \tcp*{Length of the shortest path that we have found so far}
$e_\text{mid} \leftarrow \bot$ \tcp*{Middle edge of the shortest path that we have found so far}
$u_s \gets s, u_t \gets t$ \tcp*{Vertices currently being explored in the two executions}
$\d(s, \cdot) \gets +\infty$; $\d(\cdot, t) \gets +\infty$; $\d(s, s) \gets 0$; $\d(t, t) \gets 0$\;
\bigskip
Initialize forward search from $s$ on $G$ and backward search from $t$ on $G$ with edges reversed\;
\While{neither of the two executions has terminated}{
    Alternate between relaxing one edge in the Forward and Backward algorithm\;
}
\bigskip

$uv \leftarrow e_\text{mid}$\;
$P \leftarrow$ ``shortest $su$-path according to forward execution'' + $e_\text{mid}$ + ``shortest $vt$-path according to backward execution''\;

\Return $P$\;
\bigskip

\Fn{\Fwd}{
Open $s$\;
 \While{an open vertex exists}{
        Let $u$ be the open vertex with the smallest $\d(s, u)$\; % in the forward search\;
        Close $u$\;
        Set $u_s \gets u$\;
        \If{$\d(s, u_s) + \d(u_t, t) \ge \mu$\label{line:break_condition}}{
            \textbf{terminate} the whole algorithm\;
        }
        \For{$v$ a forward neighbor of $u$} {
            \If{$v$ is not closed}{
                $\d(s, v) \gets \min(\d(s,v), \d(s,u) + \ell(uv))$\;
            }
            \If{$\d(s, u) + \ell(uv) + \d(v, t) < \mu$\footnote{In the earlier versions of this paper, this line was wrong.}}{
                \label{line:mu_update_condition}
				$\mu \gets \d(s, u) + \ell(uv) + \d(v, t)$\; 
                \label{line:mu_def}
                $e_\text{mid} \gets uv$
            }                                
        }
    }
}
\Fn{\Bwd}{ 
	Analogous to Forward algorithm with the roles of \emph{forward} and \emph{backward} flipped\;
}

	%The shortest path from $s$ to $v$, followed by edge $vw$, followed by the shortest path from $t$ to $w$ in reverse order.
\caption{Instance-Optimal Bidirectional Dijkstra's Algorithm}
\label{alg:instanceoptimal}
\end{algorithm}

\begin{theorem} \label{thm:main_theorem_weighted}
%	\cref{alg:instanceoptimal} is correct on both directed and undirected graphs. For any algorithm $A'$ that is correct with probability $2/3$, any input (directed or undirected) graph $G$, and any two vertices $s,t$, we have that the expected query complexity is $T_{\textrm{Alg\ref{alg:instanceoptimal}}}(G, s, t) = O\left( T_{A'}(G, s, t) \right)$. 
% any correct algorithm (even randomized, for which correctness means it returns a shortest path with probability $\geq 0.9$)\vasek{this should be written after the definition of the problem or after the i.o. definition} for the shortest $st$-path problem uses in expectation at least as many queries as our algorithm does, up to a constant.% This holds both on undirected and directed graphs.
\cref{alg:instanceoptimal} is an instance-optimal algorithm, under query complexity, for the shortest $st$-path problem in both directed and undirected graphs with positive weights.
\end{theorem}
We prove the theorem next. The proof is conceptually similar to that of \Cref{thm:unidirectional_dijkstra}, but needs to handle some additional issues. 

Before the proof, let us recapitulate what we need to prove. First, we need to prove that the algorithm is correct. Then, by \Cref{def:instance_optimality}, we need to show that for any algorithm $A'$ that is correct with probability $0.9$, any input (directed or undirected) graph $G$, and any two vertices $s,t$, we have that the expected query complexity is $T_{\textrm{Alg\ref{alg:instanceoptimal}}}(G, s, t) = O\left( T_{A'}(G, s, t) \right)$.

%The underlying idea is straightforward: to understand it, let us consider proving the optimality of the algorithm with respect to the number of vertices, not edges encountered, and only for deterministic algorithms. Let us also consider a hypothetical algorithm $A'$ that does not explore all the vertices in either of the two balls that \cref{alg:instanceoptimal} grows around $s$ and $t$, let us call the two such vertices $u,v$. But what if we add a very short edge between $u$ and $v$? This way, we create a new graph $G'$ with shorter distance between $s$ and $t$. Yet, the hypothetical algorithm $A'$ returns the same answer for both $G$ and $G'$ and hence it is not correct on all instances. 
%Our proof is a more careful adaptation of this argument. 

Since our \cref{alg:instanceoptimal} is slightly different than the variants of the bidirectional search that appeared before, we also need to verify its correctness. %which we do by loosely following the presentation from \cite{MITRecitation16}. %Also, note that while the claim holds directed graphs, the algorithm cannot be implemented if we do note store the reverse edges for each vertex.

\vasek{atRH: add picture}

Throughout, we use the following standard facts about the two executions of Dijkstra's algorithm.
\begin{enumerate}[(i)]
	\item \label{fact:labels} At all times $\d(s,w) \ge d(s,w)$ and $\d(w,t) \ge d(w,t)$ for every vertex $w$, with equality whenever $w$ is closed in the respective execution; moreover, $\d(s,s) = \d(t,t) = 0$ from the very beginning.
	\item \label{fact:mu} The value of $\mu$ is always the length of some (not necessarily simple) $st$-walk, hence $\mu \ge d(s,t)$.
	\item \label{fact:relaxed} In each execution, vertices are closed in order of non-decreasing distance from $s$ (resp.\ to $t$), and all edges of a closed vertex are relaxed before the next vertex is closed. Hence, at any moment, every vertex $w$ with $d(s,w) < d(s,u_s)$ is closed in the forward execution and all its edges have already been relaxed there; the analogous statement holds for the backward execution and vertices $w$ with $d(w,t) < d(u_t,t)$.
	\item \label{fact:relax_path} Let $uv$ be an edge of a shortest $st$-path. If the forward execution relaxes $uv$ at a moment when $\d(v,t) = d(v,t)$, then \crefline{line:mu_def} yields $\mu \le \d(s,u) + \ell(uv) + \d(v,t) = d(s,u) + \ell(uv) + d(v,t) = d(s,t)$, where we used that $u$ is closed. Symmetrically, if the backward execution relaxes $uv$ at a moment when $\d(s,u) = d(s,u)$, then $\mu \le d(s,t)$ afterwards.
\end{enumerate}

\begin{lemma} \label{lem:stopping}
Consider any moment at which \cref{alg:instanceoptimal} evaluates the condition on \crefline{line:break_condition}, and let $d_s = d(s,u_s)$ and $d_t = d(u_t,t)$ at this moment. If $d_s + d_t \ge d(s,t)$, then $\mu = d(s,t)$ at this moment, and in particular the algorithm terminates.
\end{lemma}
\begin{proof}
Let $\mu' = d(s,t)$ and fix a shortest $st$-path $P = (p_0, p_1, \dots, p_k)$ with $p_0 = s$ and $p_k = t$. By \cref{fact:mu} it suffices to show $\mu \le \mu'$. The algorithm is symmetric with respect to the two executions, so we may assume $d_s > 0$ (otherwise $d_t \ge \mu' > 0$ and we exchange the roles of $s$ and $t$).

Let $u = p_i$ be the last vertex on $P$ with $d(s,u) < d_s$; it exists since $d(s,s) = 0 < d_s$. By \cref{fact:relaxed}, $u$ is closed in the forward execution and all its edges have been relaxed.

If $u = t$, then also $d(s, p_{k-1}) < d(s,t) < d_s$, so $p_{k-1}$ has relaxed the edge $p_{k-1}t$ in the forward execution while $\d(t,t) = 0 = d(t,t)$, and \cref{fact:relax_path} gives $\mu \le \mu'$.

Otherwise, let $v = p_{i+1}$ be the next vertex on $P$. By the choice of $u$, $d_s \le d(s,v)$, and hence $d(v,t) = \mu' - d(s,v) \le \mu' - d_s \le d_t$. If $v = t$, then $u$ has relaxed $ut$ while $\d(t,t) = 0$ and again $\mu \le \mu'$ by \cref{fact:relax_path}. Otherwise, let $v' = p_{i+2}$. Since all weights are positive, $d(v',t) < d(v,t) \le d_t$, so by \cref{fact:relaxed} $v'$ is closed in the backward execution and all its edges have been relaxed there. Now consider the two events ``the forward execution relaxes $uv$'' and ``the backward execution relaxes $vv'$'', both of which have happened. If the first event happened later, then at that moment $\d(v,t) \le d(v',t) + \ell(vv') = d(v,t)$, so $\d(v,t) = d(v,t)$ by \cref{fact:labels}, and \cref{fact:relax_path} gives $\mu \le \mu'$. If the second event happened later, then at that moment $\d(s,v) \le d(s,u) + \ell(uv) = d(s,v)$, so $\d(s,v) = d(s,v)$, and the symmetric statement in \cref{fact:relax_path} (applied to the edge $vv'$) gives $\mu \le \mu'$.
\end{proof}

\begin{proof}[Proof of \cref{thm:main_theorem_weighted}]
\textbf{Correctness.} By correctness of Dijkstra's algorithm, the calculated shortest-path distance $\d(s, v)$ to any closed vertex $v$ equals the true distance $d(s, v)$. %Hence, all $\d(\cdot, \cdot)$ on \crefline{line:mu_def,line:break_condition} of \cref{alg:instanceoptimal} are equal to their $d(\cdot, \cdot)$ counterparts.\rh{not true} 
Let $d_s,d_t$ be the respective values of $d(s,u_s)$ and $d(u_t,t)$ when we stopped. The value of $\mu$ corresponds to the length of a (not necessarily simple) $st$-path. 
We will show it is in fact a shortest $st$-path.
Suppose first that the algorithm stopped because the forward execution ran out of open vertices. Then every vertex reachable from $s$ is closed and all its edges have been relaxed. If $t$ is unreachable, no relaxed edge can have $\d(s,u) + \ell(uv) + \d(v,t) < +\infty$, so $\mu = +\infty$ and $e_\text{mid} = \bot$, which we interpret as reporting that $t$ is unreachable. Otherwise, the predecessor $w$ of $t$ on a shortest $st$-path has relaxed the edge $wt$ while $\d(t,t) = 0$, so $\mu \le d(s,t)$ by \cref{fact:relax_path}, and $\mu = d(s,t)$ by \cref{fact:mu}. The case that the backward execution ran out of open vertices is symmetric (using $\d(s,s) = 0$).

Suppose next that the algorithm stopped on \crefline{line:break_condition}. Then $d_s + d_t = \d(s,u_s) + \d(u_t,t) \ge \mu \ge d(s,t)$ by \cref{fact:labels,fact:mu}, and \cref{lem:stopping} gives $\mu = d(s,t)$.

Finally, let $uv = e_\text{mid}$ be the edge that determined the final value of $\mu$. The returned walk consists of a shortest $su$-walk found by the forward execution (of length $\d(s,u)$, which has not changed since $u$ was closed), the edge $uv$, and a $vt$-walk of length equal to the current value of $\d(v,t)$, which is at most its value when $\mu$ was last updated. Hence the returned walk has length at most $\mu = d(s,t)$, and being an $st$-walk it has length exactly $d(s,t)$; in particular it is a shortest $st$-path.

\textbf{Instance optimality.}
We seek to prove that no algorithm for the shortest $st$-path problem can use fewer queries than \Cref{alg:instanceoptimal} by more than a constant factor. In the following argument, we argue that no algorithm can be faster for the problem of computing the $s$-$t$ distance (instead of the shortest $st$-path problem). This implies the claim because any algorithm that computes the path can be easily modified to also calculate the distance without increasing its complexity. 

Let $E_s$ and $E_t$ be the sets of the edges that our algorithm explored using the two executions of Dijkstra's algorithm. Note that we have $|E_t| \le |E_s| \le |E_t| + 1$ by the definition of \cref{alg:instanceoptimal}.

First, note that the query complexity of \cref{alg:instanceoptimal} is proportional to $|E_s| + |E_t| \le 2|E_s|$. 
On the other hand, suppose that an algorithm $A$ explores at most $|E_t|/4$ edges in expectation; here, we say that $A$ explores an edge if it queries it via the operations \textbf{Edge}, \textbf{Inedge} or \textbf{Outedge} described in \cref{sec:preliminaries} from any of its endpoints. We will prove that under this assumption, the algorithm $A$ is incorrect with a probability of at least $1/2$.

	To see this, note that in both $E_s$ and $E_t$, there exists an edge that is accessed with probability at most $1/4$ by $A$: Otherwise,  by the linearity of expectation, the expected number of edges visited by $A$ would be more than $|E_t|/4$ which we assume is not the case. We call these two edges $e_1 = u_1v_1$ and $e_2 = u_2v_2$ (note that it could happen that $e_1 = e_2$). In the directed case, $u_1$ is the tail of $e_1$ and $v_2$ is the head of $e_2$. In the undirected case, we without loss of generality assume that $d(s,u_1) \leq d(s,v_1)$ and $d(v_2,t) \leq d(u_2,t)$. %the same holds by the definition of \cref{alg:instanceoptimal} in the directed case.\rh{directed nepravda}
% In the case of undirected graphs, we assume without loss of generality that $d(s,u_1) \leq d(s,v_1)$ and $d(v_2,t) \leq d(u_2,t)$.

We next claim that the shortest $st$-path has its length strictly larger than $ d(s,u_1) + d(v_2,t)$.\vasek{for the next version, we need to split the proof into claims this should be a claim its used twice. also if we dont simplify proofs we should have more numbered equations and refer to them when we use them}
	To see this, let $a$ be the vertex that was closed in the forward execution when $e_1$ was relaxed there, and let $b$ be the vertex that was closed in the backward execution when $e_2$ was relaxed there. In the directed case $a = u_1$ and $b = v_2$; in the undirected case $a \in \{u_1, v_1\}$ and $b \in \{u_2, v_2\}$. In either case, our choice of $u_1$ and $v_2$ gives $d(s,u_1) \le d(s,a)$ and $d(v_2,t) \le d(b,t)$.
	Now consider the later of the two evaluations of \crefline{line:break_condition} that took place when $a$ and $b$ were closed. At this moment, both $a$ and $b$ are closed, so $d(s,a) \le d(s,u_s)$ and $d(b,t) \le d(u_t,t)$ by \cref{fact:relaxed}. Moreover, the algorithm did not terminate at this evaluation, since afterwards it relaxed the edges of $a$ (respectively of $b$), in particular $e_1$ (respectively $e_2$). Hence, by \cref{lem:stopping}, $d(s,u_s) + d(u_t,t) < d(s,t)$, and we conclude that $d(s,u_1) + d(v_2,t) \le d(s,a) + d(b,t) \le d(s,u_s) + d(u_t,t) < d(s,t)$.

Next, let us define $\delta = d(s,t) - d(s,u_1) - d(v_2,t)$; note that $\delta > 0$ by our previous claim. 
We now construct a graph $G'$ with a different distance between $s$ and $t$ than in $G$ with the property that $A$ returns the same answer on both $G$ and $G'$ with large probability. 

% \textbf{Case 1: $v_1 \neq u_2$}
% Consider first the case when $v_1 \neq u_2$. In this case, w
We start with the case $e_1 \neq e_2$. In this case, we define $G'$ by replacing $e_1,e_2$ in $G$ by two edges $u_1v_2$ and $u_2v_1$ where the first edge has length $\delta'$ for arbitrary $0 < \delta' < \delta$ and the second edge has arbitrary weight. Note that the degrees of all vertices are the same in the two graphs, and they also have the same number of vertices, although $G'$ may contain parallel edges and self-loops even if $G$ did not. 

Furthermore, note that the only queries that would return different answers on the two graphs are the neighborhood queries that would access either of the edges $e_1,e_2$ on $G$ but they would access the edges $u_1v_2,u_2v_1$ on $G'$. 
However, we assumed that with probability at least $1/2$, $A$ accesses neither $e_1$ nor $e_2$ when run on $G$. 
This also implies that $A$ does not access the edges $u_1v_2, u_2v_1$ on $G'$ with that probability, since until those edges are accessed, the runs of $A$ on $G$ and $G'$ perform the exact same queries (assuming they use the same source of randomness). 

We conclude that with probability at least $1/2$, $A$ returns the same answer on both graphs $G$ and $G'$. 
We will next argue that the correct answers are different on the two graphs. 

To see this, we will verify that $d_{G'}(s,u_1) \leq d_{G}(s,u_1)$ and $d_{G'}(v_2,t) \leq d_{G}(v_2,t)$.\footnote{In fact, it holds that those distances are the same, but we will not need to argue this.} We show the first inequality since the second inequality can be proven in the same way. Specifically, we will prove that no shortest $su_1$-path uses $e_1$ or $e_2$, which implies the desired inequality. We will assume that $G$ is undirected as the directed case is similar and, in fact, easier.\rh{fix properly}

%\global\long\def\thesame{%
We start with $e_1$. We recall our assumption that $d(s,u_1) \leq d(s,v_1)$. Together with the fact that all edge weights are positive, this implies that no shortest $su_1$-path can use the edge $e_1$ in the direction from $v_1$ to $u_1$. The edge is clearly not used in the opposite direction. 

We continue with $e_2$. If it lay on the shortest path from $s$ to $u_1$, we would have $d(s, v_2) \le d(s, u_1)$. Moreover, there would be a path from $s$ to $t$ of length at most $d(s, v_2) + d(v_2, t) \le d(s, u_1) + d(v_2, t)$. But we have proven that the distance $d(s,t)$ is strictly larger than this quantity. 

We conclude that $d_{G'}(s,u_1) \leq d_{G}(s,u_1)$. 
%}\thesame
Finally, we recall the equality
$d_G(s,t) = d_G(s,u_1) + d_G(v_2,t) + \delta$ to conclude that
\[d_{G'}(s,t) \leq d_{G'}(s,u_1) + d_{G'}(v_2,t) + \delta' < d_{G}(s,u_1) + d_{G}(v_2,t) + \delta = d_G(s,t).\]
In particular, $d_G(s,t) \not= d_{G'}(s,t)$ and we conclude that $A$ is incorrect with probability at least $1/2$ on either $G$ or $G'$, a contradiction. 
% \textbf{Case 2: $v_1 = u_2$}
% In the case $v_1 = u_2$, we define the graph $G'$ to have the same edge set as $G$ but we redefine the weights of $e_1,e_2$ to $\delta/3$. Just like in the above case, the only queries that return different answers on $G$ and $G'$ are those that would access either $e_1$ or $e_2$, which $A$ does not do with probability at least $1/2$. 
% Analogously to the first case, we can verify that $d_G(s,u_1) \leq d_{G'}(s,u_1)$ and $d_G(v_2,t) \leq d_{G'}(v_2,t)$. We conclude that  
% \[
% 	d_{G'}(s,t) \leq d_{G'}(s,u_1) + d_{G'}(v_2,t) + 2 \cdot\delta/3 \leq d_{G}(s,u_1) + d_{G}(v_2,t) + 2\delta/3 < d_G(s,t) \,.
% 	\qedhere
% \]
% Again, we have $d_{G'}(s,t) \not= d_G(s,t)$ and we conclude that $A$ is incorrect with probability at least $3/4$, a contradiction. 

Finally, it remains to consider the easier case that $e_1 = e_2$. In this case, the only difference between $G$ and $G'$ is that we set the length of $e_1 = e_2$ in $G'$ to be $\delta'$ for arbitrary $0 < \delta' < \delta$. We then have 
\[
d_{G'}(s,t) \leq d(s,u_1) + \delta' + d(v_2,t) < d(s,t) \,.
\]
By an analogous argument as above, it can be argued that $A$ is incorrect with probability at least $1/2$ on either $G$ or $G'$.
\end{proof}

\subsection{Remarks regarding our setup}
\label{subsec:remarks}

We make a few remarks regarding our setup and our proof of instance optimality. 

\paragraph{Flexibility in our algorithm}
First, we note that there is flexibility in \cref{alg:instanceoptimal}: instead of alternating between edges, it suffices to make sure that the total number of edges explored in either execution is the same, up to constant factors. 

% \begin{corollary}
% Consider a modification of the bidirectional Dijkstra's algorithm that alternates between exploring arbitrary number of edges in one execution, and arbitrary number in the other. If at all points in time the numbers of edges explored in the two executions differ by at most a constant factor, then the algorithm is instance-optimal up to a constant factor for the shortest st-path problem.
% \end{corollary}
% \begin{proof}
% 	We consider the run of the (original) bidirectional Dijkstra's algorithm. Let $k$ be the number of edges it explores. By the alternating nature of the algorithm, each execution has explored at least $\lfloor k/2\rfloor$ edges.
% 	Now consider the modified algorithm. At the time it stops, we have $\min(|E_s|, |E_t|) \le \lfloor k/2\rfloor$, since when $\min(|E_s|, |E_t|) = \lfloor k/2\rfloor$, we have explored all the edges that the original algorithm did.

% The number of explored edges in both executions is the same up to a constant in both algorithms, which means that the runs of both algorithms explored the same number of edges up to a constant. Since the original algorithm was instance-optimal, so is the algorithm from this corollary.
% \end{proof}

\paragraph{Zero-length edges}
Next, we remark that it is crucial that we assume that edge weights are positive reals. In the case when we allow zero-length edges, instance optimality is no longer possible: Just imagine a large graph where all the edges have weight zero. An algorithm taylored to a particular graph $G$ can start with an advice constituting of a path between $s$ and $t$; it suffices to walk along that path and confirm that $d(s,t) = 0$ while algorithms without such advice have to dutifully explore $G$. 

Formally, nonexistence of instance-optimal algorithms in the case when $0$-weight edges are allowed can be proven by considering the graph $G_0$ that consists of two complete binary trees rooted at $s$ and $t$. Moreover, select random $i$ and connect by an edge the $i$-th vertices on the last layer of the two binary trees. One can notice that any algorithm has to make $\Omega(n)$ queries in expectation on a randomly sampled $G_0$, while for any sample, there exists a fixed fast algorithm tailored to it that only makes $O(\log n)$ queries before finishing.\footnote{This construction implicitly assumes that vertices have unique identifiers -- otherwise we cannot construct the suitable advice. However, the need for unique identifiers can be easily removed by adding to each vertex a neighbor (center of a star) with degree equal to the identifier. }

\paragraph{Positive lower bound on edge weights}
It is also crucial that as weights, we allow arbitrarily small positive numbers: In the proof, we define a certain quantity $\delta$ and construct a graph $G'$ with edge weight $0 < \delta' < \delta$. This assumption is again necessary; we will see in \cref{sec:unweighted} that even on unweighted graphs (where there is a minimum edge weight of $1$), instance optimality is achievable only approximately. We note that even in the unweighted setting, we could prove that bidirectional search is instance-optimal, if we changed the task of finding the $st$-distance (or one shortest $st$-path) to the task of finding \emph{all} shortest $st$-paths. 

% It may also seem inconsequential whether the algorithm uses floating-point numbers instead of arbitrary-precision reals. This is not quite true, although perhaps it is not far from true.\rh{are we hedging here? otherwise I'd drop, sounds a bit clunky IMO} Our argument relies on the possibility of finding a smaller non-negative weight for any given weight. Indeed, this is also why zero weights cause issues. Similar issues occur with floating-point numbers, since they have a smallest positive value that can be represented. However, our proof still works on instances where none of the weights are equal to this smallest representable positive value, and our algorithm is instance-optimal for those instances.\rh{not true now, we need all weights to be $>2\mu$ now, because of the self-loop case.}

\paragraph{Accessing directed edges from both endpoints}
In our proof, it is also crucial that each (directed) edge can be accessed not only from the source, but also from its target. 
If it could be accessed only from the source, one could see that the standard Dijkstra's algorithm is instance-optimal (under the assumption of strictly positive real edge-weights). 

\paragraph{Computing $s$-$t$ distance vs computing a shortest $st$-path}
Our proof actually shows a somewhat stronger statement. Namely, it show that no algorithm for the easier problem of computing the $s$-$t$ distance can be faster on any instance than \Cref{alg:instanceoptimal}, which solves the harder problem of computing an actual shortest path.

\section{Approximate Instance Optimality in Unweighted Graphs}
\label{sec:unweighted}

In this section, we prove that bidirectional search is approximately instance-optimal also on unweighted graphs. See \cref{fig:lower_bound} for the intuition behind the $O(\Delta)$ term that we have to lose in the analysis. We will use the same algorithm, \cref{alg:instanceoptimal}, for the proof. We note that in the unweighted case, no heap is needed for the implementation and the query complexity and the time complexity are equal. We also call this adapted algorithm \emph{the bidirectional BFS} algorithm.\footnote{We note that \cref{alg:instanceoptimal} could be further sped up if we also terminate the search once we encounter an edge $uv$ where $u$ was seen in the execution from $s$ and $v$ from $t$. However, this more practical algorithm is not $o(\Delta)$-instance-optimal, so we focus on the conceptually simpler \cref{alg:instanceoptimal}. }

\begin{theorem}
\label{thm:main_unweighted}
The bidirectional BFS algorithm for the shortest $st$-path problem in unweighted graphs is instance-optimal, under both query and time complexity, up to the factor of $O(\Delta)$.
\end{theorem}
We note that the following proof mostly follows the proof of \cref{thm:main_theorem_weighted}. 
\begin{proof}
% \emph{Correctness.} Consider the two vertices that were currently being explored in the respective executions of BFS; let $d_s, d_t$ be the distances to these vertices from $s$ and $t$, respectively. Note that the returned path's length is then $d_s + d_t + 1$. Let $V_s$ and $V_t$ be the sets of vertices seen from the respective executions of BFS (we say a vertex is \emph{seen} if it is $s,t$, or if its id was returned by a query). For the sake of contradiction, assume the shortest path has length $\leq d_s + d_t$. This means that it has to contain an edge $uv$ between $V_s$ and $V_t$, with either  $d(s,u) < d_s$ or $d(v,t) < d_t$; assume without loss of generality the first is the case. We know that the vertex $u$ has been explored. If, when that happened, the execution from $t$ has already seen $v$, then we would have stopped at that point, which is a contradiction. If this were not the case, $v$ would belong to $V_s$ and not $V_t$, which is again a contradiction. \vasek{why do we even argue correcntess here? i would make the algorithm special case of the weighted one so that this is clear}

% \emph{Instance-optimality up to $O(\Delta)$.}
We have already shown correctness in \Cref{thm:main_theorem_weighted}. It thus remains to prove instance-optimality. As in the proof of \Cref{thm:main_theorem_weighted}, we prove that no faster algorithm can be correct for the easier problem of computing the distance from $s$ to $t$. As in that proof, this implies the claim.
We note that in the unweighted setting, once a node $u$ is opened with some value $\d(s,u)$, we know that $\d(s,u) = d(s,u)$. We may thus use $\d$ and $d$ values interchangeably. 

Consider the step in \cref{alg:instanceoptimal} in which we redefined the value $\mu$ to its final value (i.e., the length of the shortest path) for the first time. We will use $\mu^0 = d(s,t)$ to denote this final value. Without loss of generality, we assume that this happened during the forward run. 

We use $u^0,v^0,w^0$ to denote the following three nodes. The node $u^0$ is the node that is being explored during the step in the forward run that defined $\mu^0$. The node $v^0$ is its neighbor that was used to define $\mu^0$ as $\mu^0 = \d(s, u^0) + \ell(u^0, v^0) + \d(v^0, t)$. Finally, $w^0$ is the node such that $v^0$ was opened by $w^0$ in the backward run. 

We will use $d_s, d_t$ to denote the distances $d(s, u^0)$ and $d(w^0, t)$. Note that we have 
\[
\mu^0 = d_s + 2 + d_t
\]
where the additive term $+2$ is for the two edges $u^0v^0$ and $v^0w^0$. 

We let $E_s$ be the set of edges that are outgoing from any vertex $u$ with $d(s,u) \leq d_s$ and let $E_t$ be the set of edges ingoing to any $u$ such that $d(u,t) \leq d_t$; in the undirected case, replace ``outgoing/ingoing'' by ``adjacent''. 

%\footnote{Note the difference from the definitions in the proof of \Cref{thm:main_theorem_weighted}, where the sets $E_s, E_t$ contained all explored edges; now there can be additional edges incident to vertices of distance $d_s, d_t$ that our algorithm explored but that are not contained in $E_s, E_t$. }

For the sake of contradiction, let us assume the existence of an algorithm $A'$ that, in expectation, explores at most $\min(|E_s|,|E_t|)/4$ edges. Then, using the linearity of expectation, we conclude that there are two edges $e_1 = u_1v_1 \in E_s$ and $e_2 = u_2v_2 \in E_t$, such that with probability at least $1/2$, neither edge is queried by $A'$. %In the undirected case, choose $u_1,v_2$ so that $d(s, u_1) \leq d_s $ and $d(v_2, t) \leq d_t $. %note that there is no lower bound on $d(s,u_2), d(v_2,t)$ due to the edges being directed. 
In the undirected case, we assume without loss of generality that $d(s,u_1) \leq d(s,v_1)$ and similarly that $d(v_2,t) \leq d(u_2,t)$. This implies $d(s, u_1) \le d_s $ and $d(v_2, t) \le d_t $.

We observe that $e_1 \neq e_2$. Otherwise, we would have that
\[
d(s,t) \leq d(s,u_1) + 1 + d(v_1,t) \leq d_s + 1 + d_t
\]
which would be in contradiction with $d(s,t) = d_s + 2 + d_t$.

We define a new multigraph $G'$ by replacing the edges $u_1v_1,v_2u_2$ by $u_1u_2$ and $v_1v_2$. Note that $G'$ may contain self-loops and parallel edges even if $G$ did not. 
We observe that the only queries that distinguish $G$ from $G'$ correspond to accessing $u_1v_1$ or $u_2v_2$ in $G$. Thus, $A$ behaves differently on $G$ and $G'$ with probability at most $1/2$. We argue below that $d_G(s,t) \neq d_{G'}(s,t)$; this implies that the success probability of $A$ is at most $1/2$ which in turn proves that any correct algorithm has to have query complexity at least $\Omega(\min(|E_s|, |E_t|))$ on $G$.

We now need to argue that $d_G(s,t) \not= d_{G'}(s,t)$. 
We first claim that $d_{G'}(s,u_1) \leq d_{G}(s,u_1)$ and $d_{G'}(v_2,t) \leq d_{G}(v_2,t)$; we will argue only for the first inequality as the second proof is the same. 
Specifically, we will prove that no shortest $su_1$-path uses $e_1$ and $e_2$ which implies our claim since then any shortest $su_1$-path in $G$ also exists in $G'$. We will assume that $G$ is undirected as the directed case is similar and, in fact, easier.

\rh{fix those two paragraphs}
We start with $e_1$. We recall our assumption that $d(s,u_1) \leq d(s,v_1)$. This implies that any $su_1$-path that uses the edge $e_1$ in the direction from $v_1$ to $u_1$ has length at least $d(s,v_1) + 1 \ge d(s,u_1) + 1$ and is thus not a shortest path. The edge is clearly not used in the opposite direction. 

We continue with $e_2$. If it lay on a shortest path from $s$ to $u_1$, we would have $d(s, v_2) \le d(s, u_1) - 1 \le d_s - 1$. Moreover, we have $d(v_2, t) \le d_t$. This implies $d(s, t) \le d(s, v_2) + d(v_2, t) \le d_s + d_t - 1$ which is in contradiction with the fact that $d(s,t) = d_s + d_t + 2$.

We conclude that $d_{G'}(s,u_1) \leq d_{G}(s,u_1)$ and analogously $d_{G'}(v_2,t) \leq d_{G}(v_2,t)$. We can now use this to compute that
\begin{align*}
  d_{G'}(s,t) 
  &\leq d_{G'}(s,u_1) + d_{G'}(v_2,t) + 1 \\
  &\leq d_{G}(s,u_1) + d_{G}(v_2,t) + 1 \\
&\le d_s + d_t + 1\,.  
\end{align*}
On the other hand, we have by definition that 
\[
d_G(s,t) = \mu = d_s + d_t + 2
\]
and we thus have $d_G(s,t) \not= d_{G'}(s,t)$. 

It remains to argue that the time complexity of the algorithm is at most $O(\Delta \cdot \min(|E_s|, |E_t|))$. 
To see this, we first observe that the forward search in our algorithm closes at most one vertex of distance at least $d_s + 2$ from $s$. At the point in time when such a node $u$ is first closed, we have $\d(s, u_s) \ge d_s + 2$. However, we would also have $\d(u_t, t) \ge d_t$ since the node $w^0$ was closed by the backward algorithm, by definition of $w^0$. Thus, we have $\d(s,u_s) + \d(u_t,t) \ge d_s + d_t + 2 = \mu^0$; this however triggers the condition on \crefline{line:break_condition} and the algorithm terminates. We conclude that our algorithm closes at most $|E_s| + 1$ vertices and hence it explores $O(|E_s| \cdot \Delta)$ edges. An analogous argument can be made for the backward search. We conclude that our algorithm makes $O(\min(|E_s|, |E_t|) \cdot \Delta)$ queries, as needed. 
\end{proof}

\subsection{Lower bound} \label{sec:lower_bound}

Next, we prove that the factor of $O(\Delta)$ in \cref{thm:main_unweighted} cannot be improved. We prove this in a slightly more general setup where the set of allowed positive weights, $\mathcal{W}$, is bounded away from zero. 

\begin{theorem}
	\label{thm:lower-bound-minima}
Assume the shortest $st$-path problem when the allowed graph weights come from a set $\mathcal{W}$ with $\nu \coloneq \min(\mathcal W) > 0$ and we restrict the class of input graphs to those of degree at most $\Delta$. 
Then there is no algorithm that is instance optimal, under both query and time complexity, for the problem up to a factor of $o(\Delta)$.\footnote{We do not regard $\Delta$ as a constant in the $o$ notation. }
\end{theorem}

\begin{figure}[t]
  \centering
  \includegraphics[width=0.8\textwidth]{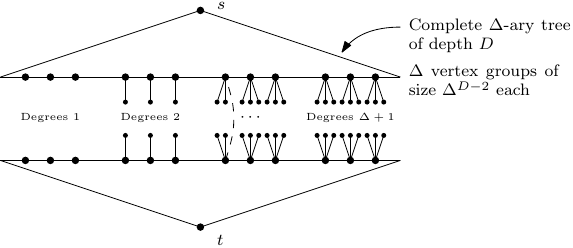}
	\caption{The lower bound instance from \cref{thm:lower-bound-minima}: There are $\Delta$ groups of vertices and one of them contains a hidden edge connecting $s$ to $t$. If we know which group contains the hidden edge, we only need to search one of the groups, otherwise we have to search all of them. 
 Note that this counterexample crucially uses that the hidden edge has the smallest possible weight; otherwise, all the groups need to be searched as they could contain an edge with even smaller weight. }
  \label{fig:lower_bound}
\end{figure}

\begin{proof}
	Consider the following graph $G_1$ illustrated in \Cref{fig:lower_bound}. All edges in $G_1$ have weight $\nu$. To construct $G_1$,  we first construct a tree as follows. Consider a complete $\Delta$-ary tree of depth $D$ rooted in $s$ -- for convenience, we use $\Delta$ to be the arity of the tree, resulting in maximum degree of $\Delta+1$. 
 Remove some of the vertices in the last layer so that the first $\Delta^{D-2}$ vertices on the penultimate layer have degree $\Delta + 1$, the next $\Delta^{D-2}$ vertices have degree $\Delta$, and so on. We then construct a second, isomorphic tree rooted in $t$. Next, for some value $i$, we consider the $i$-th edge in the last layer in both of these trees, we remove the two edges, and connect the corresponding two vertices on the penultimate layer with a new edge.

	If $i$ is chosen uniformly randomly, any algorithm with a constant success probability has to visit, in expectation, a constant fraction of all vertices on the last level before it finds the connecting edge. Hence, its expected complexity is $\Omega(\Delta^D)$. On the other hand, for any fixed graph, we will now describe an algorithm that is correct and runs in time $O(\Delta^{D-1})$. 
 
	Consider an instance of the graph $G_1$ where the connecting edge connects vertices with degree $k$. Next, consider an algorithm that first runs a bidirectional BFS taylored to $G_1$ and $k$ that works as follows. 
 The algorithm explores the vertices from $s$ and $t$ until the penultimate layer. Then, it explores the vertices with degree $k$. If it finds an edge of length $\nu$ between the two components containing $s$ and $t$, the algorithm finishes. 
 If such an edge is not found, or if the algorithm finds that the input graph is not $G_1$, the algorithm disposes of the run and runs any correct algorithm such as Dijkstra's algorithm to compute the result. 

We claim that this algorithm is correct on all inputs. This follows from the fact that once the two subtrees are explored, any shortest path has to have length at least $(2D-1) \cdot \nu$; it is thus not necessary to explore the rest of the graph when a path of this length is found. 

Next, we consider the time complexity of the algorithm on $G_1$. The algorithm explores all vertices except those on the last layer where it only opens some of them. On the last layer, it explores vertices adjacent to vertices with degrees $k$. The number of vertices on all layers except the last is $O(\Delta^{D-1})$. The number of vertices seen on the last layer is at most $\Delta \cdot \Delta^{D-2} = \Delta^{D-1}$. Overall, the complexity is thus $O(\Delta^{D-1})$.
%
%This means that for any graph in this graph class exists a correct algorithm with complexity $O(n/\Delta)$. But the worst-case complexity is $\Theta(n)$, meaning that no algorithm can be instance-optimal up to $o(\Delta)$.
\end{proof}

We remark that while bidirectional BFS is not instance-optimal by the above theorem, it is still instance-optimal when compared only to algorithms that do not assume that the graph is unit-weight. Said slightly more formally, bidirectional BFS is instance-optimal against algorithms that are correct on all graphs where the edge weights come from the set $\{1/2, 1\}$ (or indeed $\{\delta, 1\}$ for any $\delta < 1$). This follows from the proof of \cref{thm:main_theorem_weighted}.

\section*{Open problems}
It would be interesting to see whether the \Astar{} algorithm or some of its bidirectional variants \cite{holte2017} allow similarly strong guarantees as bidirectional search. 

%Our proof of \cref{thm:main_informal} works for multigraphs that allow parallel edges and self-loops. While we believe that the possibility of self-loops can be avoided, we are not sure whether the same holds for parallel edges: We believe that whether \cref{thm:main_informal} holds in the setting of simple graphs is an interesting open question. 

\section*{Acknowledgments}
BH, RH, VR, and JT were partially funded by the Ministry of Education and Science of Bulgaria's support for INSAIT as part of the Bulgarian National Roadmap for Research Infrastructure. Part of this work was done while JT and VR were working at INSAIT and RH was visiting INSAIT.
BH and RH were partially funded through the European Research Council (ERC) under the European Union's Horizon 2020 research and innovation program (ERC grant agreement 949272).
VR was partially funded through the European Research Council (ERC) under the European Union's Horizon 2020 research and innovation program (ERC grant agreement 853109)
RH and JT were supported by the VILLUM Foundation grant 54451. 
Part of this work was done while JT was working at and RH was visiting BARC at the University of Copenhagen. RH would like to thank Rasmus Pagh for hosting him there.
RT's research at Princeton was partially supported by a gift from Microsoft.  Part of this work was done during RT's visits to INSAIT and to the Simons Institute for the Theory of Computing.

%We have shown that the bidirectional \Astar{} algorithm is not instance-optimal for rather trivial reasons. Could we prove some weaker optimality results that are at the same time significantly stronger than those in previous work such as \cite{shaham2019}?

%In \Astar{} there also seems to be a tradeoff between the number of queries performed on the graph and the running time.\vasek{because you invest more time into better heuristic? aha you are talking about some bidirectional astar algo?} The reason is that if one wants to minimize the number of edge queries, in each of the two executions of \Astar{}, we likely want to use the distance to the set of vertices explored in the other execution as part of the heuristic. But this could result in performing many queries on the advice function which lowerbounds the distances in the graph. Providing deeper understanding of this tradeoff would be very interesting.\vasek{i dont get this, i rewrote the first problem to something i understand but i dont get the second one. we dont really discuss bidirectional astar enough here to talk about it in such detail imho. i tried to merge into previous problem so that we can drop this one}

\printbibliography
\end{document}